\documentclass{article}
\usepackage{nohyperref}
\usepackage{amsmath}
\usepackage{amssymb}
\usepackage{graphicx}
\usepackage{natbib}
\usepackage{url}
\numberwithin{equation}{section}
\usepackage{multicol}
\usepackage[left=0.5in,right=0.5in,bottom=0.75in,top=0.75in]{geometry}
\usepackage{ctable}
\setupctable{pos=ht!}

\renewcommand{\d}{\mathbf{d}}

\providecommand{\I}{\mathbf{I}}

\providecommand{\Q}{\mathbf{Q}}
\renewcommand{\r}{\mathbf{r}}

\renewcommand{\v}{\mathbf{v}}

\providecommand{\W}{\mathbf{W}}
\providecommand{\x}{\mathbf{x}}
\providecommand{\X}{\mathbf{X}}
\providecommand{\y}{\mathbf{y}}
\providecommand{\z}{\mathbf{z}}

\providecommand{\zero}{\mathbf{0}}

\providecommand{\lam}{\lambda}

\providecommand{\bb}{\boldsymbol{\beta}}
\providecommand{\bh}{\hat{\beta}}
\providecommand{\bbh}{\widehat{\boldsymbol{\beta}}}

\providecommand{\be}{\boldsymbol{\eta}}

\providecommand{\mh}{\hat{\mu}}

\providecommand{\bL}{\boldsymbol{\Lambda}}

\providecommand{\bp}{\boldsymbol{\pi}}

\providecommand{\tbb}{\tilde{\boldsymbol{\beta}}}
\providecommand{\bbj}{\boldsymbol{\beta}_{-j}}

\providecommand{\rr}{\tilde{\mathbf{r}}}

\providecommand{\Xj}{\mathbf{X}_{-j}}
\providecommand{\XX}{\widetilde{\mathbf{X}}}

\providecommand{\yy}{\tilde{\mathbf{y}}}

\renewcommand{\Pr}{\mathbb{P}}

\providecommand{\Ex}{\mathbb{E}}

\providecommand{\SE}{\textrm{SE}}

\providecommand{\abs}[1]{\left\lvert#1\right\rvert}
\providecommand{\norm}[1]{\left\lVert#1\right\rVert}

\providecommand{\al}[2]{\begin{align}\label{#1}#2\end{align}}
\providecommand{\as}[1]{\begin{align*}#1\end{align*}}

\renewcommand{\abstract}[1]{
 \centerline{
 \begin{minipage}{0.7\linewidth}
 \hrule
 \vskip 0.1in
  \begin{center}
    {\bf Abstract}
  \end{center}
  #1
 \vskip 0.1in
 \hrule
 \end{minipage}}
 \vskip 0.3in} 

\usepackage{titling}
\pretitle{\begin{center} \LARGE \bf}
\posttitle{\par\end{center}\vskip 0.5em}

\usepackage{algorithm}
\usepackage{algpseudocode}
\algdef{SxnE}[FOR]{For}{EndFor}[1]{\algorithmicfor\ #1 }

\usepackage{amsthm}
\newtheorem{prop}{Proposition}
\newtheorem{lemma}{Lemma}

\title{Group descent algorithms for nonconvex penalized linear and logistic regression models with grouped predictors}
\author{Patrick Breheny\\Department of Biostatistics\\University of Iowa \and Jian Huang\\Department of Statistics and Actuarial Sciences\\Department of Biostatistics\\University of Iowa}
\date{\today}

\begin{document}

\maketitle

\abstract{Penalized regression is an attractive framework for variable selection problems. Often, variables possess a grouping structure, and the relevant selection problem is that of selecting groups, not individual variables.  The group lasso has been proposed as a way of extending the ideas of the lasso to the problem of group selection.  Nonconvex penalties such as SCAD and MCP have been proposed and shown to have several advantages over the lasso; these penalties may also be extended to the group selection problem, giving rise to group SCAD and group MCP methods.  Here, we describe algorithms for fitting these models stably and efficiently.  In addition, we present simulation results and real data examples comparing and contrasting the statistical properties of these methods.}

\section{Introduction}

In regression modeling, explanatory variables can often be thought of as grouped.   To represent a categorical variable, we may introduce a group of indicator functions.  To allow flexible modeling of the effect of a continuous variable, we may introduce a series of basis functions.  Or the variables may simply be grouped because the analyst considers them to be similar in some way, or because a scientific understanding of the problem implies that a group of covariates will have similar effects.

Taking this grouping information into account in the modeling process should improve both the interpretability and the accuracy of the model.  These gains are likely to be particularly important in high-dimensional settings where sparsity and variable selection play important roles in estimation accuracy.

Penalized likelihood methods for coefficient estimation and variable selection have become widespread since the original proposal of the lasso \citep{Tibshirani1996}.  Building off of earlier work by \citet{Bakin1999}, \citet{Yuan2006a} extended the ideas of penalized regression to the problem of grouped covariates.  Rather than penalizing individual covariates, Yuan and Lin proposed penalizing norms of groups of coefficients, and called their method the group lasso.

The group lasso, however, suffers from the same drawbacks as the lasso.  Namely, it is not consistent with respect to variable selection and tends to over-shrink large coefficients. These shortcomings arise because the rate of penalization of the group lasso does not change with the magnitude of the group coefficients, which leads to biased estimates of large coefficients. To compensate for the over-shrinkage, the group lasso tends to include spurious coefficients into the model.

The smoothly clipped absolute deviation (SCAD) penalty and the minimax concave penalty (MCP) were developed in an effort to achieve what the lasso could not: simultaneous selection consistency and asymptotic unbiasedness \citep{Fan2001,Zhang2010}.  This achievement is known as the {\em oracle property}, so named because it implies that the model is asymptotically equivalent to the fit of a maximum likelihood model in which the identities of the truly nonzero coefficients are known in advance.  These properties extend to group SCAD and group MCP models, as shown in \citet{Wang2008a} and \citet{Huang2012}.

However, group SCAD and group MCP have not been widely used or studied in comparison with the group lasso, largely due to a lack of efficient and publicly available algorithms for fitting these models.  Published articles on the group SCAD \citep{Wang2007,Wang2008a} have used a local quadratic approximation for fitting these models.  The local quadratic approximation was originally proposed by \citet{Fan2001} to fit SCAD models.  However, by relying on a quadratic approximation, the approach is incapable of producing naturally sparse estimates, and therefore cannot take advantage of the computational benefits provided by sparsity.  This, combined with the fact that solving the local quadratic approximation problem requires the repeated factorization of large matrices, makes the algorithm very inefficient for fitting large regression problems.  \citet{Zou2008} proposed a local linear approximation for fitting SCAD models and demonstrated its superior efficiency to local quadratic approximations.  This algorithm was further improved upon by \citet{Breheny2011}, who demonstrated how a coordinate descent approach may used to fit SCAD and MCP models in a very efficient manner capable of scaling up to deal with very large problems.

Here, we show how the approach of \citet{Breheny2011} may be extended to fit group SCAD and group MCP models.  We demonstrate that this algorithm is very fast and stable, and we provide a publicly available implementation in the \texttt{grpreg} package, (\url{http://cran.r-project.org/web/packages/grpreg/index.html}).  In addition, we provide examples involving both simulated and real data which demonstrate the potential advantages of group SCAD and group MCP over the group lasso.

\section{Group descent algorithms}

We consider models in which the relationship between the explanatory variables, which consist of $J$ non-overlapping groups, and the outcome is specified in terms of a linear predictor $\be$:
\al{eq:lin-pred}{\be = \beta_0 + \sum_{j=1}^J \X_j\bb_j,}
where $\X_j$ is the portion of the design matrix formed by the predictors in the $j$th group and the vector $\bb_j$ consists of the associated regression coefficients.  Letting $K_j$ denote the number of members in group $j$, $\X_j$ is an $n \times K_j$ matrix with elements $(x_{ijk})$, the value of $k$th covariate in the $j$th group for the $i$th subject.  Covariates that do not belong to any group may be thought of as a group of one.

The problem of interest involves estimating a vector of coefficients $\bb$ using a loss function $L$ which quantifies the discrepancy between $y_i$ and $\eta_i$ combined with a penalty that encourages sparsity and prevents overfitting; specifically, we estimate $\bb$ by minimizing
\al{eq:generic}{Q(\bb) = L(\bb|\y,\X) + \sum_{j=1}^J p_{\lam}\left(\norm{\bb_j}\right),}
where $p(\cdot)$ is a penalty function applied to the Euclidean norm ($L_2$ norm) of $\bb_j$.  The penalty is indexed by a regularization parameter $\lam$, which controls the tradeoff between loss and penalty.  It is not necessary for $\lam$ to be the same for each group; {\em i.e.}, we may consider a collection of regularity parameters $\{\lam_j\}$.  For example, in practice there are often variables known to be related to the outcome and therefore which we do not wish to include in selection or penalization.  The above framework and algorithms which follow may be easily extended to include such variables by setting $\lam_j=0$.

In this section, we primarily focus on linear regression, where $\Ex(\y)=\be$ and $L$ is the least squares loss function, but take up the issue of logistic regression in Section~\ref{Sec:logistic}, where $L$ arises from the binomial likelihood.  For linear regression, we assume without loss of generality that all variables have been centered to have zero mean; in this case, $\bh_0=0$ and may be ignored.

\subsection{Orthonormalization}
\label{Sec:ortho}

The algorithm for fitting group lasso models originally proposed by \citet{Yuan2006a} requires the groups $\X_j$ to be orthonormal, as does the extension to logistic regression proposed in \citet{Meier2008}.  It is somewhat unclear from these papers, however, whether orthonormality is a necessary precondition for applying the group lasso, and if not, how one should go about applying the group lasso to non-orthonormal groups.

This question was explored in a recent work by \citet{Simon2011a}.  In that paper, the authors provide a number of compelling arguments, both theoretical and empirical, that for the group lasso, proper standardization involves orthonormalizing the groups prior to penalizing their $L_2$ norms.  In particular, they demonstrate that the resulting variable selection procedure is closely connected to uniformly most powerful invariant testing.

As we demonstrate in this article, orthonormalizing the groups also produces tremendous advantages in terms of developing algorithms to fit group penalized models.  In addition to greatly reducing the computational burden associated with fitting group lasso models, orthonormalization also leads to straightforward extensions for fitting group SCAD and group MCP models.  Importantly, as we will see, this orthonormalization can be accomplished without loss of generality since the resulting solutions can be transformed back to the original scale after fitting the model.  Thus, it is not necessary for an analyst fitting these group penalization models to have orthonormal groups or to worry about issues of orthonormality when applying these algorithms or using our {\tt grpreg} software.

Taking the singular value decomposition of the Gram matrix of the $j$th group, we have
\as{\frac{1}{n}\X_j^T\X_j = \Q_j\bL_j\Q_j^T,}
where $\bL_j$ is a diagonal matrix containing the eigenvalues of $n^{-1}\X_j^T\X_j$ and $\Q_j$ is an orthonormal matrix of its eigenvectors.  Now, we may construct a linear transformation $\XX_j=\X_j\Q_j\bL_j^{-1/2}$ with the following properties:
\al{eq:ortho1}{\frac{1}{n}\XX_j^T\XX_j &= \I\\
\XX_j\widetilde{\bb}_j &= \X_j(\Q_j\bL_j^{-1/2}\widetilde{\bb}_j)\label{eq:ortho2},}
where $\I$ is the identity matrix and $\widetilde{\bb}_j$ is the solution to \eqref{eq:generic} on the orthonormalized scale.  In other words, if we have the solution to the orthonormalized problem, we may easily transform back to the original problem with $\bb_j = \Q_j\bL_j^{-1/2}\widetilde{\bb}_j$.  This procedure is not terribly expensive from a computational standpoint, as the decompositions are being applied only to the groups, not the entire design matrix, and the inverses are of course trivial to compute because $\bL_j$ is diagonal.  Furthermore, the decompositions need only be computed once initially, not with every iteration.

Note that this procedure may be applied even when $\X_j$ is not full-rank by omitting the zero eigenvalues and their associated eigenvectors.  In this case, $\Q_j$ is a $K_j \times r_j$ matrix and $\bL_j$ is an $r_j \times r_j$ matrix, where $r_j$ denotes the rank of $\X_j$.  Given these modifications, $\bL_j$ is invertible and $\XX_j=\X_j\Q_j\bL_j^{-1/2}$ still possesses properties \eqref{eq:ortho1} and \eqref{eq:ortho2}.  Note, however, that $\XX_j$ now contains only $r_j$ columns and by extension, $\widetilde{\bb}_j$ contains only $r_j$ elements.  Thus, we avoid the problem of incomplete rank by fitting the model in a lower-dimensional parameter space, then transforming back to the original dimensions (note that applying the reverse transformation results in a $\bb_j$ with appropriate dimension $K_j$).  In our experience, it is not uncommon for non-full-rank groups to arise in problems for which one wishes to use group penalized models, and the SVD approach we describe here handles such cases naturally.  For example, in the special case where two members of a group are identical, $\x_{jk}=\x_{j\ell}$, the approach ensures that their coefficients, $\beta_{jk}$ and $\beta_{j\ell}$, will be equal for all values of $\lam$.  

For the remainder of this article, we will assume that this process has been applied and that the model fitting algorithms we describe are being applied to orthonormal groups (by ``orthonormal groups'', we mean groups for which $n^{-1}\X_j^T\X_j=\I$, not that groups $\X_j$ and $\X_k$ are orthogonal to each other).  Consequently, we drop the tildes on $\XX$ and $\widetilde{\bb}$ in subsequent sections.  

It is worth noting that penalizing the norm of $\widetilde{\bb}_j$ is not equivalent to penalizing the norm of the original coefficients.  As pointed out in \citet{Simon2011a} and \citet{Huang2012},
\as{\norm{\widetilde{\bb}_j} &= \sqrt{\frac{1}{n}\bb_j^T\X_j^T\X_j\bb_j}\\
 &\propto \norm{\be_j},}
where $\be_j = \X_j\bb_j$.  In other words, orthonormalizing the groups is equivalent to applying an $L_2$ penalty on the scale of the linear predictor.  The idea of penalizing on the scale of the linear predictors is also explored in \citet{Ravikumar2009}.  The two penalties are equivalent in the non-grouped case, provided that the standard normalization $n^{-1}\sum_i x_{ijk}^2=1$ has been applied.  However, for grouped regression models, this normalization at the coefficient level is inadequate; orthonormalization at the group level is appropriate.

\subsection{Group lasso}
\label{Sec:group-lasso}

In this section, we describe the group lasso and algorithms for fitting group lasso models.  The group lasso estimator, originally proposed by \citet{Yuan2006a}, is defined as the value $\bbh$ that minimizes
\al{eq:gl}{Q(\bb) = \frac{1}{2n}\norm{\y-\X\bb}^2 + \lam\sum_j\sqrt{K_j}\norm{\bb_j}.}
The idea behind the penalty is to apply a lasso penalty to the Euclidean ($L_2$) norm of each group, thereby encouraging sparsity and variable selection at the group level.  The solution $\bbh$ has the property that if group $j$ is selected, then $\beta_{jk} \ne 0$ for all $k$, otherwise $\beta_{jk} = 0$ for all $k$.  The correlation between $\X_j$ and the residuals will tend to be larger if group $j$ has more elements; the presence of the $\sqrt{K_j}$ term in the penalty compensates for this, thereby normalizing across groups of different sizes.  As discussed in \citet{Simon2011a}, this results in variable selection which is roughly equivalent to the uniformly most powerful invariant test for inclusion of the $j$th group.  In what follows, we will absorb the $\sqrt{K_j}$ term into $\lam$ and use $\lam_j = \lam\sqrt{K_j}$.

\citet{Yuan2006a} also propose an algorithm which they base on the ``shooting algorithm'' of \citet{Fu1998}.  Here, we refer to this type of algorithm as a ``group descent'' algorithm.  The idea behind the algorithm is the same as that of coordinate descent algorithms \citep{Friedman2007,Wu2008}, with the modification that the optimization of \eqref{eq:gl} takes place repeatedly with respect to a group $\bb_j$ rather than an individual coordinate $\beta_j$.

Below, we present the group descent algorithm for solving \eqref{eq:gl} to obtain the group lasso estimator.  The algorithm is essentially the same as Yuan and Lin's, although (a) we make explicit its generalization to the case of non-orthonormal groups using the approach described in Section~\ref{Sec:ortho}, (b) we restate the algorithm to more clearly illustrate the connections with coordinate descent algorithms, and (c) we employ techniques developed in the coordinate descent literature to speed up the implementation of the algorithm considerably.  As we will see, this presentation of the algorithm also makes clear how to easily extend it to fit group SCAD and group MCP models in the following sections.

\begin{figure*}
 \centering
 \includegraphics[width=0.7\linewidth]{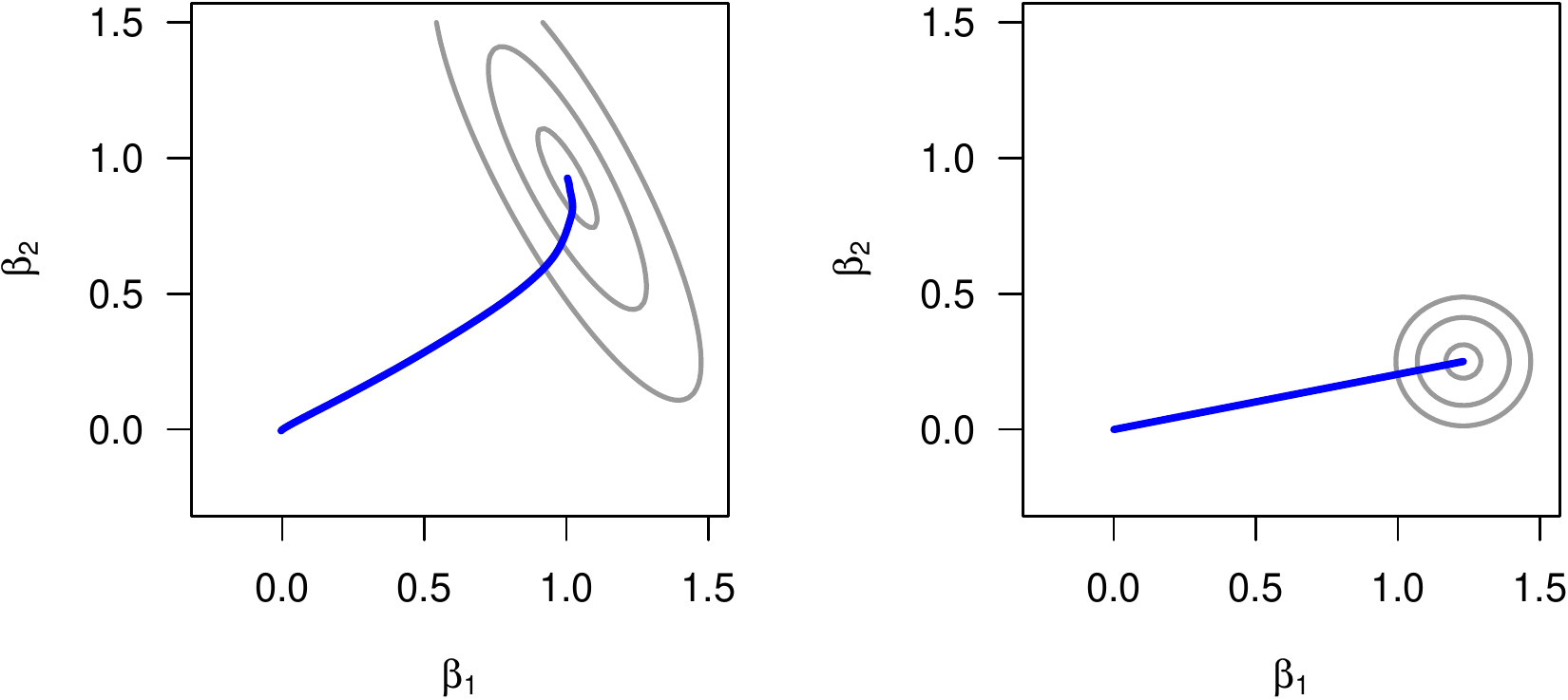}
 \caption{\label{Fig:geom} The impact of orthonormalization on the solution to the group lasso.  Contour lines for the likelihood (least squares) surface are drawn, centered around the OLS solution, as well as the solution path for the group lasso as $\lam$ goes from 0 to $\infty$.  {\em Left}: Non-orthonormal $\X$.  {\em Right}: Orthonormal $\X$.}
\end{figure*}

We begin by noting that the subdifferential \citep{Bertsekas1999} of $Q$ with respect to $\bb_j$ is given by
\al{eq:score}{\partial Q(\bb_j) = \begin{cases}
  -\z_j + \bb_j + \lam_j \bb_j/\norm{\bb_j} & \text{if $\bb_j=\zero$}\\
  -\z_j + \lam_j \v & \text{if $\bb_j=\zero$}\\
\end{cases}}
where $\z_j = \X_j^T(\y-\Xj\bbj)$ is the least squares solution, $\Xj$ is the portion of the design matrix that remains after $\X_j$ has been excluded, $\bbj$ are its associated regression coefficients, and $\v$ is any vector satisfying $\norm{\v} \leq 1$.  The main implication of \eqref{eq:score} is that, by orthonormalizing the groups, $\X_j$ drops out of the equation and the multivariate problem of optimizing with respect to $\bb_j$ is reduced to a univariate problem, as the solution must lie on the line segment joining $\zero$ and $\z_j$.  The geometry of this problem is illustrated in Figure~\ref{Fig:geom}.  As the figure illustrates, orthonormalization renders the direction of $\bbh_j$ invariant with respect to $\lam$, thereby enabling us to break the solution down into two components: determining the direction of $\bb_j$ and determining its length.

Furthermore, determining the length of $\bb_j$ is equivalent to solving a one-dimensional lasso problem, which has a simple, closed-form solution given by the soft-thresholding operator \citep{Donoho1994}:
\al{eq:soft}{
  S(z,\lambda) = \begin{cases}
    z - \lambda & \text{if } z > \lambda \\
    0 & \text{if } \abs{z} \leq \lambda \\
    z + \lambda & \text{if } z < -\lambda \text{.}
  \end{cases}
}
With a slight abuse of notation, we extend this definition to a vector-valued argument $\z$ as follows:
\al{eq:mv-soft}{S(\z,\lam) = S(\norm{z},\lam) \frac{\z}{\norm{\z}},}
where $\z/\norm{\z}$ is the unit vector in the direction of $\z$.  In other words, $S(\z,\lam)$ acts on a vector $\z$ by shortening it towards $\zero$ by an amount $\lam$, and if the length of $\z$ is less than $\lam$, the vector is shortened all the way to $\zero$.

The multivariate soft-thresholding operator is the solution to the single-group group lasso, just as the univariate soft-thresholding operator is the solution to the single-variable lasso problem.  This leads to Algorithm~\ref{Alg:gd-lasso}, which is exactly the same as the coordinate descent algorithm described in \citet{Friedman2010} for fitting lasso-penalized regression models, only with multivariate soft-thresholding replacing univariate soft-thresholding.  Both algorithms are essentially modified backfitting algorithms, with soft-thresholding replacing the usual least squares updating.

\begin{algorithm}
\caption{\label{Alg:gd-lasso} Group descent algorithm for the group lasso}
  \begin{algorithmic}
    \Repeat
    \For{$j = 1, 2, \ldots, J$}
      \State $\z_j = \X_j^T\r + \bb_j$
      \State $\bb_j' \gets S\left(\z_j,\lam_j\right)$
      \State $\r' \gets \r - \X_j^T(\bb_j'-\bb_j)$
    \EndFor
    \Until convergence
  \end{algorithmic}
\end{algorithm}

In Algorithm~\ref{Alg:gd-lasso}, $\bb_j$ refers to the current ({\em i.e.}, most recently updated) value of coefficients in the $j$th group prior to the execution of the for loop; during the loop, $\bb_j$ is updated to $\bb_j'$.  The same notation is applied to $\r$, where $\r$ refers to the residuals: $\r = \y - \sum_j\X_j\bb_j$.  The ``$\gets$'' refers to the fact that $\bb_j$ and $\r$ are being continually updated; at convergence, $\bbh$ consists of the final updates $\{\bb_j\}$.  The expression $\z_j = \X_j^T\r + \bb_j$ is derived from
\al{eq:partial-residuals}{\z_j = \X_j^T(\y-\Xj\bbj) = \X_j^T\r + \bb_j;}
the algorithm is implemented in this manner because it is more efficient computationally to update $\r$ than to repeatedly calculate the partial residuals $\y-\Xj\bbj$, especially in high dimensions.

The computational efficiency of Algorithm~\ref{Alg:gd-lasso} is clear: no complicated numerical optimization steps or matrix factorizations or inversions are required, only a small number of simple arithmetic operations.  This efficiency is possible only because the groups $\{\X_j\}$ are made to be orthonormal prior to model fitting.  Without this initial orthonormalization, we cannot obtain the simple closed-form solution \eqref{eq:mv-soft}, and the updating steps required to fit the group lasso become considerably more complicated, as in \citet{Friedman2010a}, \citet{Foygel2010}, and \citet{Puig2011}.

\subsection{Group MCP and group SCAD}

We have just seen how the group lasso may be viewed as applying the lasso/soft-thresholding operator to the length of each group.  Not only does this formulation lead to a very efficient algorithm, it also makes it clear how to extend other univariate penalties to the group setting.  Here, we focus on two popular alternative univariate penalties to the lasso: SCAD, the smoothly clipped absolute deviation penalty \citep{Fan2001} and MCP, the minimax concave penalty \citep{Zhang2010}.

The two penalties are similar in motivation, definition, and performance.  The penalties are defined on $[0,\infty)$ for $\lam > 0$ as follows, and plotted on the left side of Figure~\ref{Fig:shapes}:
\al{eq:scad}{
  \textrm{SCAD:} \quad p_{\lambda, \gamma}(\theta) &= \begin{cases} \lambda \theta & \text{if } \theta \leq \lambda \\
    \frac{\gamma\lambda\theta-0.5(\theta^2+\lambda^2)}{\gamma-1} & \text{if } \lambda < \theta \leq \gamma\lambda \\
    \frac{\lambda^2(\gamma^2-1)}{2(\gamma-1)} & \text{if } \theta > \gamma\lambda \end{cases}
  \\ \label{eq:mcp}
  \textrm{MCP:} \quad p_{\lambda,\gamma}(\theta) &= \begin{cases} \lambda \theta - \frac{\theta^2}{2\gamma} & \text{if } \theta \leq \gamma\lambda \\
    \frac{1}{2} \gamma\lambda^2 & \text{if } \theta > \gamma\lambda \end{cases}.
}
To have a well-defined minimum, we must have $\gamma>1$ for MCP and $\gamma > 2$ for SCAD.  Although originally proposed for univariate penalized regression, these penalties may be extended to the grouped-variable selection problem by substituting \eqref{eq:scad} and \eqref{eq:mcp} into \eqref{eq:generic}, as has been proposed in \citet{Wang2007} and discussed in \citet{Huang2012}.  We refer to these penalized regression models as the Group SCAD and Group MCP methods, respectively.

The rationale behind the penalties can be understood by considering their derivatives, which appear in the middle panel of Figure~\ref{Fig:shapes}.  MCP and SCAD begin by applying the same rate of penalization as the lasso, but continuously relax that penalization until the point at which $\theta = \gamma\lambda$, where the rate of penalization has fallen all the way to 0.  The aim of both penalties is to achieve the variable selection properties of the lasso, but to introduce less bias towards zero among the true nonzero coefficients.  The only difference between the two is that MCP reduces the rate of penalization immediately, while SCAD remains flat for a while before moving towards zero.

\begin{figure}
 \centering
 \includegraphics[width=\linewidth]{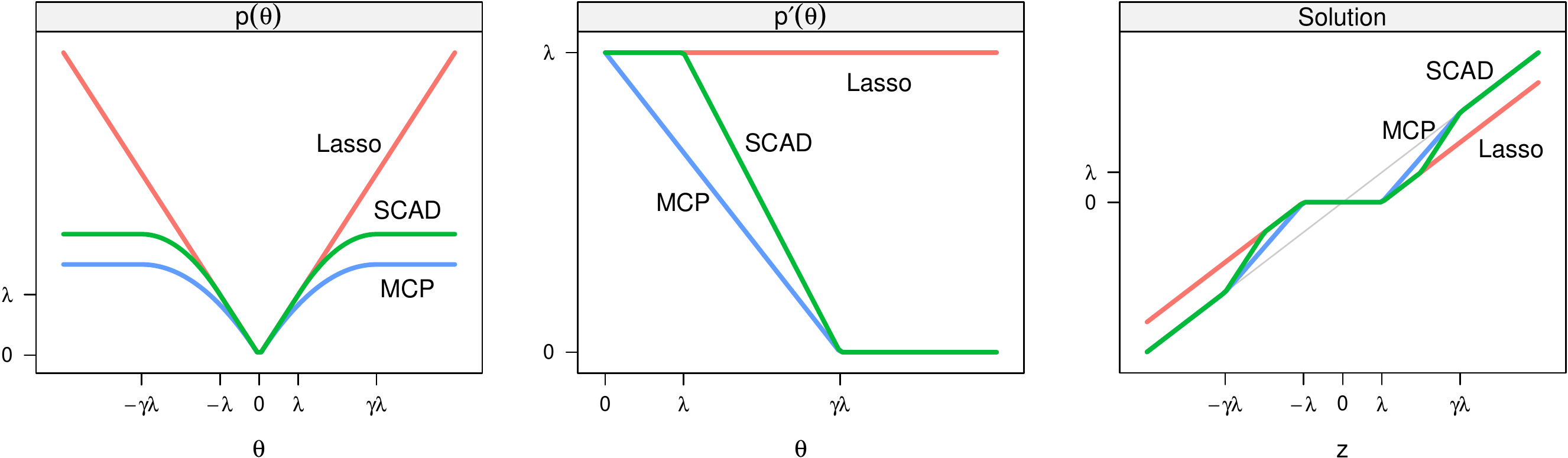}
 \caption{\label{Fig:shapes} Lasso, SCAD, and MCP penalty functions, derivatives, and univariate solutions.  The panel on the left plots the penalties themselves, the middle panel plots the first derivative of the penalty, and the right panel plots the univariate solutions as a function of the ordinary least squares estimate.  The light gray line in the rightmost plot is the identity line.  Note that none of the penalties are differentiable at $\beta_j=0$.}
\end{figure}

The rationale behind the penalties can also be understood by considering their univariate solutions.  Consider the simple linear regression of $\y$ upon $\x$, with unpenalized least squares solution $z=n^{-1}\x^T\y$.  For this simple linear regression problem, the MCP and SCAD estimators have the following closed forms:
\al{eq:mcp-sol}{
  \bh &= F(z,\lambda,\gamma) = \begin{cases} \frac{S(z,\lambda)}{1-1/\gamma} & \text{if } \abs{z} \leq \gamma\lambda \\
    z & \text{if } \abs{z} > \gamma\lambda \text{,} \end{cases} \\ \label{eq:scad-sol}
  \bh &= F_S(z,\lambda,\gamma) = \begin{cases} S(z,\lambda) & \text{if } \abs{z} \leq 2\lambda \\
    \frac{S(z,\gamma\lambda/({\gamma-1}))}{1-1/(\gamma-1)} & \text{if } 2\lambda < \abs{z} \leq \gamma\lambda \\
    z & \text{if } \abs{z} > \gamma\lambda \text{.} \end{cases}}
Noting that $S(z,\lambda)$ is the univariate solution to the lasso, we can observe by comparison that MCP and SCAD scale the lasso solution upwards toward the unpenalized solution by an amount that depends on $\gamma$.  For both MCP and SCAD, when $\abs{z} > \gamma\lam$, the solution is scaled up fully to the unpenalized least squares solution.  These solutions are plotted in the right panel of Figure~\ref{Fig:shapes}; the figure illustrates how the solutions, as a function of $z$, make a smooth transition between the lasso and least squares solutions.  This transition is essential to the oracle property described in the introduction.

As $\gamma \to \infty$, the MCP and lasso solutions are identical.  As $\gamma \to 1$, the MCP solution becomes the hard thresholding estimate $zI_{\abs{z}>\lambda}$.  Thus, in the univariate sense, the MCP produces the ``firm shrinkage'' estimator of \citet{Gao1997}; hence the $F(\cdot)$ notation.  The SCAD solutions are similar, of course, but not identical, and thus involve a ``SCAD-modified firm thresholding'' operator which we denote $F_S(\cdot)$.  In particular, the SCAD solutions also have soft-thresholding as the limiting case when $\gamma \to \infty$, but do not have hard thresholding as the limiting case when $\gamma \to 2$.

We extend these two firm-thresholding operators to multivariate arguments as in \eqref{eq:mv-soft}, with $F(\cdot)$ or $F_S(\cdot)$ taking the place of $S(\cdot)$, and note that $F(\z_j,\lam,\gamma)$ and $F_S(\z_j,\lam,\gamma)$ optimize the objective functions for Group MCP and Group SCAD, respectively, with respect to $\bb_j$.  An illustration of the nature of these estimators is given in Figure~\ref{Fig:paths}.  We note the following: (1) All estimators carry out group selection, in the sense that, for any value of $\lam$, the coefficients belonging to a group are either wholly included or wholly excluded from the model.  (2) The group MCP and group SCAD methods eliminate some of the bias towards zero introduced by the group lasso.  In particular, at $\lam \approx 0.2$, they produce the same estimates as a least squares regression model including only the nonzero covariates (the ``oracle'' model).  (3) Group MCP makes a smoother transition from $\zero$ to the unpenalized solutions than group SCAD.  This is the ``minimax'' aspect of the penalty.  Any other penalty that makes the transition between these two extremes must have some region ({\em e.g.} $\lam \in [0.7, 0.5]$ for group SCAD) in which its solutions are changing more abruptly than those of group MCP.

\begin{figure}
 \centering
 \includegraphics[width=\linewidth]{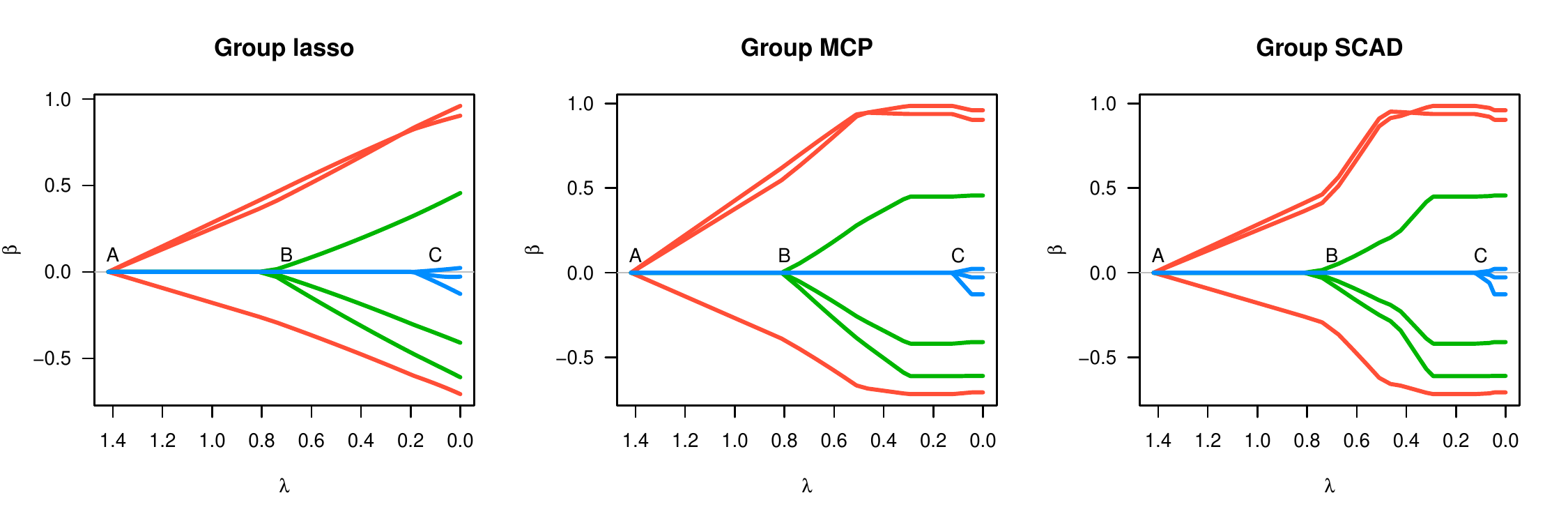}
 \caption{\label{Fig:paths} Representative solution paths for the group lasso, group MCP, and group SCAD methods.  In the generating model, groups A and B have nonzero coefficients and while those belonging to group C are zero.}
\end{figure}

It is straightforward to extend Algorithm~\ref{Alg:gd-lasso} to fit Group SCAD and Group MCP models; all that is needed is to replace the multivariate soft-thresholding update with a multivariate firm-thresholding update.  The group updates for all three methods are listed below:
\as{
  \textrm{Group lasso}: \bb_j &\gets S(\z_j,\lam_j) = S(\norm{\z_j},\lam_j) \frac{\z_j}{\norm{\z_j}}\\
  \textrm{Group MCP}: \bb_j &\gets F(\z_j,\lam_j, \gamma) = F(\norm{\z_j},\lam_j, \gamma) \frac{\z_j}{\norm{\z_j}}\\
  \textrm{Group SCAD}: \bb_j &\gets F(\z_j,\lam_j, \gamma)  = F_S(\norm{\z_j},\lam_j, \gamma) \frac{\z_j}{\norm{\z_j}};
}
recall that $\lam_j = \lam\sqrt{K_j}$.  Note that all three updates are simple, closed-form expressions.  Furthermore, as each update minimizes the objective function with respect to $\bb_j$, the algorithms possess the descent property, meaning that they decrease the objective function with every iteration.  This, along with the fact that $Q$ is a strictly convex function with respect to $\bb_j$ (see Lemma 1 in the Appendix for details) leads to attractive convergence properties, which we now state formally.

\begin{prop}
\label{Prop:linear}
Let $\bb^{(m)}$ denote the value of the fitted regression coefficient at the end of iteration $m$.  At every iteration of the proposed group descent algorithms for linear regression models involving group lasso, group MCP, or group SCAD penalties,
\as{Q(\bb^{(m+1)}) \leq Q(\bb^{(m)}).}
Furthermore, every limit point of the sequence $\{\bb^{(1)}, \bb^{(2)}, \ldots\}$ is a stationary point of $Q$.
\end{prop}
For the group lasso penalty, the objective function being minimized is convex, and thus the above proposition establishes convergence to the global minimum.  For group MCP and group SCAD, whose objective function is a sum of convex and nonconvex components, convergence to a local minimum is possible.

A similar algorithm was explored in \citet{She2012}, who proposed the same updating steps as above, although without an initial orthonormalization.  Interestingly, She showed that even without orthonormal groups, the updating steps described above still produce a sequence $\bb^{(m)}$ converging to a stationary point of the objective function.  Without orthonormal groups, however, the updates are not exact group-wise solutions.  In other words, in the single-group case, our approach (with an initial orthonormalization) produces exact solutions in one step, whereas the approach in \citet{She2012} requires multiple iterations to converge to the same solution.  This leads to a considerable loss of efficiency, as we will see in Section~\ref{Sec:time}.

In conclusion, the algorithms we present here for fitting group lasso, group MCP, and group SCAD models are both fast and stable.  We examine the empirical running time of these algorithms in Section~\ref{Sec:time}.

\subsection{Logistic regression}
\label{Sec:logistic}

It is possible to extend the algorithms described above to fit group-penalized logistic regression models as well, where the loss function is the negative log-likelihood of a binomial distribution:
\as{L(\be) = \frac{1}{n}\sum_i L_i(\eta_i) = -\frac{1}{n}\sum_i \log \Pr(y_i|\eta_i).}
Recall, however, that the simple, closed form solutions of the previous sections were possible only with orthonormalization.  The iteratively reweighted least squares (IRLS) algorithm typically used to fit generalized linear models (GLMs) introduces a $n^{-1}\X^T\W\X$ term into the score equation \eqref{eq:score}, where $\W$ is an $n \times n$ diagonal matrix of weights.  Because $n^{-1}\X^T\W\X \neq \I$, the group lasso, group MCP, and group SCAD solutions will lack the simple closed forms of the previous section.

However, we may preserve the sphericity of the likelihood surface (Figure~\ref{Fig:geom}) through the application of a majorization-minimization (MM) approach \citep{Lange2000,Hunter2004}.  In the context of penalized logistic regression, this approach was proposed by \citet{Krishnapuram2005}, who referred to it as a {\em bound optimization algorithm}.  The application of the method depends on the ability to bound the Hessian of the loss function with respect to the linear predictor $\be$.  Let $v=\max_i \sup_\eta \{\nabla^2 L_i(\eta)\}$, so that $v\I-\nabla^2 L(\be)$ is a positive semidefinite matrix at all points $\be$.  For logistic regression, where
\as{\pi_i = \Pr(Y_i=1|\eta_i) = \frac{e^{\eta_i}}{1+e^{\eta_i}},}
we may easily obtain $v=1/4$, since $\nabla^2 L_i(\eta)=\pi(1-\pi)$.

Bounding the loss function in this manner allows us to define
\as{\tilde{L}(\be|\be^*) = L(\be^*) + (\be-\be^*)^T\nabla L(\be^*) + \frac{v}{2}(\be-\be^*)^T(\be-\be^*)}
such that the function $\tilde{L}(\be|\be^*)$ has the following two properties:
\as{
  \tilde{L}(\be^*|\be^*) &= L(\be^*)\\
  \tilde{L}(\be|\be^*) &\geq L(\be).
}
Thus, $\tilde{L}(\be|\be^*)$ is a majorizing function of $L(\be)$.  The theory underlying MM algorithms then ensures that Algorithm~\ref{Alg:mm}, which consists of alternately executing the majorizing step and the minimization steps,  will retain the descent property of the previous sections, which we formally state below.

\begin{prop}
\label{Prop:logistic}
Let $\bb^{(m)}$ denote the value of the fitted regression coefficient at the end of iteration $m$.  At every iteration of the proposed group descent algorithms for logistic regression involving group lasso, group MCP, or group SCAD penalties,
\as{Q(\bb^{(m+1)}) \leq Q(\bb^{(m)}).}
Furthermore, provided that no elements of $\bb$ tend to $\pm\infty$, every limit point of the sequence $\{\bb^{(1)}, \bb^{(2)}, \ldots\}$ is a stationary point of $Q$.
\end{prop}
As with linear regression, this result implies convergence to a global minimum for the group lasso, but allows convergence to local minima for group MCP and group SCAD.  Note, however, that unlike linear regression, in logistic regression maximum likelihood estimates can occur at $\pm\infty$ (this is often referred to as {\em complete separation}).  In practice, this is not a concern for large values of $\lam$, but saturation of the model will certainly occur when $p > n$ and $\lam$ is small.  Our implementation in \texttt{grpreg} terminates the path-fitting algorithm if saturation is detected, based on a check of whether $>99\%$ of the null deviance has been explained by the model.

Writing $\tilde{L}(\be|\be^*)$ in terms of $\bb$, we have
\as{\tilde{L}(\bb) \propto \frac{v}{2n}(\yy - \X\bb)^T(\yy - \X\bb),}
where $\yy = \be^* + (\y-\bp)/v$ is the pseudo-response vector.  Thus, the gradient of $\tilde{L}(\be|\be^*)$ with respect to $\bb_j$ is given by
\al{eq:score-mm}{\nabla \tilde{L}(\bb_j) = -v\z_j + v\bb_j,}
where, as before, $\z_j = \X_j^T(\yy-\Xj\bbj)$ is the vector of partial (pseudo-) residuals for $\bb_j$.

\begin{algorithm}
\caption{\label{Alg:mm} Group descent algorithm for logistic regression with a group lasso penalty}
  \begin{algorithmic}
    \Repeat
      \State $\be \gets \X\bb$
      \State $\bp \gets \{e^{\eta_i}/(1+e^{\eta_i})\}_{i=1}^n$
      \State $\rr \gets (\y-\bp)/v$
      \For{$j = 1, 2, \ldots, J$}
        \State $\z_j = \X_j^T\rr + \bb_j$
        \State $\bb_j' \gets S\left(v\z_j,\lam_j\right)/v$
        \State $\rr' \gets \rr - \X_j^T(\bb_j'-\bb_j)$
    \EndFor
    \Until convergence
  \end{algorithmic}
\end{algorithm}

The presence of the scalar $v$ in the score equations affects the updating equations; however, as the majorized loss function remains spherical with respect to $\bb$, the updating equations still have simple, closed form solutions:
\as{
  \textrm{Group lasso}: \bb_j &\gets \frac{1}{v} S(v\z_j,\lam_j) = \frac{1}{v} S(v\norm{\z_j},\lam_j) \frac{\z_j}{\norm{\z_j}}\\
  \textrm{Group MCP}: \bb_j &\gets \frac{1}{v} F(v\z_j,\lam_j, \gamma) = \frac{1}{v} F(v\norm{\z_j},\lam_j, \gamma) \frac{\z_j}{\norm{\z_j}}\\
  \textrm{Group SCAD}: \bb_j &\gets \frac{1}{v} F(v\z_j,\lam_j, \gamma)  = \frac{1}{v} F_S(v\norm{\z_j},\lam_j, \gamma) \frac{\z_j}{\norm{\z_j}}
}
Algorithm~\ref{Alg:mm} is presented for the group lasso, but is easily modified to fit group MCP and group SCAD models by substituting the appropriate expression into the updating step for $\bb_j$.

Note that Proposition~\ref{Prop:logistic} does not necessarily follow for other generalized linear models, as the Hessian matrices for other exponential families are typically unbounded.  One possibility is to set $v \gets \max_i \{\nabla^2 L_i(\eta_i^*)\}$ at the beginning of each iteration as a pseudo-upper bound.  As this is not an actual upper bound, an algorithm based on it is not guaranteed to possess the descent property.  Nevertheless, the approach seems to work well in practice.  The authors have tested the approach on the group-penalized Poisson regression models and did not observe any problems with convergence, although we have not studied these models as extensively as the logistic regression model.

\subsection{Path-fitting algorithm}
\label{Sec:path}

The above algorithms are presented from the perspective of fitting a penalized regression model for a single value of $\lam$.  Usually, we are interested in obtaining $\bbh$ for a range of $\lam$ values, and then choosing among those models using either cross-validation or some form of information criterion.  The regularization parameter $\lam$ may vary from a maximum value $\lambda_{\max}$ at which all penalized coefficients are 0 down to $\lambda=0$ or to a minimum value $\lambda_{\min}$ beyond which the model becomes excessively large.  When the objective function is strictly convex, the estimated coefficients vary continuously with $\lam \in [\lam_{\min}, \lam_{\max}]$ and produce a path of solutions regularized by $\lam$.  An examples of such a path may be seen in Figure~\ref{Fig:paths}.

Algorithms~\ref{Alg:gd-lasso} and \ref{Alg:mm} are iterative and require initial values; the fact that $\bbh=\zero$ at $\lam_{\max}$ provides an efficient approach to choosing those initial values.  Group lasso, group MCP, and group SCAD all have the same value of $\lam_{\max}$; namely, $\lam_{\max}=\max_j\{\norm{\z_j}\}$ for linear regression or $\lam_{\max}=\max_j\{v\norm{\z_j}\}$ for logistic regression, where the $\{\z_j\}$ are computed with respect to the intercept-only model (or, if unpenalized covariates are present, with respect to the residuals of the fitted model including all of the unpenalized covariates).  Thus, by starting at $\lam_{\max}$ where $\bbh=0$ is available in closed form and proceeding towards $\lambda_{\min}$, using $\bbh$ from the previous value of $\lam$ as the initial value for the next value of $\lam$, we can ensure that the initial values will never be far from the solution, a helpful property often referred to as ``warm starts'' in the path-fitting literature.

\section{Algorithm efficiency}
\label{Sec:time}

Here, we briefly comment on the efficiency of the proposed algorithms.  Regardless of the penalty chosen, the most computationally intensive steps in Algorithm~\ref{Alg:gd-lasso} are the calculation of the inner products $\X_j^T\r$ and $\X_j^T(\bb_j'-\bb_j)$, each of which requires $O(nK_j)$ operations.  Thus, one full pass over all the groups requires $O(2np)$ operations.  The fact that this approach scales linearly in $p$ allows it to be efficiently applied to high-dimensional problems.

Of course, the entire time required to fit the model depends on the number of iterations, which in turn depends on the data and on $\lam$.  Broadly speaking, the dominant factor influencing the number of iterations is the number of nonzero groups at that value of $\lam$, since no iteration is required to solve for groups that remain fixed at zero.  Consequently, when fitting a regularized path, a disproportionate amount of time is spent at the least sparse end of the path, where $\lam$ is small.

Table~\ref{Tab:time} compares our implementation (\texttt{grpreg}) with two other publicly available \texttt{R} packages for fitting group lasso models over increasingly large data sets: the \texttt{grplasso} package \citep{Meier2008} and the \texttt{standGL} package \citep{Simon2011a}.  We note that (a) the \texttt{grpreg} implementation appears uniformly more efficient than the others, and (b) that group MCP and group SCAD tend to be slightly faster than group lasso.  Presumably, this is because their solution paths tend to be more sparse.

\ctable[
 caption={Comparison of {\tt grpreg} with other publicly available group lasso packages.  The median time, over 100 independent data sets, required to fit the entire solution path over a grid of 100 $\lam$ values is reported in seconds.  Each group consisted of 10 variables; thus $p$ ranges over $\{10,100,1000\}$ across the columns.},
 label=Tab:time]{lllrrr}{
 \tnote[a]{$\SE \leq 0.02$}
 \tnote[b]{$\SE \leq 0.1$}
 \tnote[c]{$\SE \leq 3$}}{\FL
 & & & n=50 & n=500 & n=5000 \\
 & & & J=1\tmark[a]  & J=10\tmark[b]  & J=100\tmark[c]  \\
\midrule
Linear & \texttt{grpreg} & Group lasso & 0.01 & 0.1 & 18 \\
regression & \texttt{grpreg} & Group MCP & $< 0.01$ & 0.1 & 11 \\
 & \texttt{grpreg} & Group SCAD & $< 0.01$ & 0.1 & 11 \\
 & \texttt{grplasso} & Group lasso & 0.17 & 1.9 & 61 \\
 & \texttt{standGL} & Group lasso & 0.04 & 1.3 & 153 \\[0.2cm]
Logistic & \texttt{grpreg} & Group lasso & 0.01 & 0.2 & 105 \\
regression & \texttt{grpreg} & Group MCP & 0.01 & 0.1 & 45 \\
 & \texttt{grpreg} & Group SCAD & 0.01 & 0.1 & 78 \\
 & \texttt{grplasso} & Group lasso & 0.52 & 5.1 & 202 \\
 & \texttt{standGL} & Group lasso & 0.40 & 6.4 & 400 \LL}

It is worth noting that all three of these packages can handle $p > n$ problems; however, for the purposes of timing, we chose to restrict our attention to problems in which the entire path can be computed.  Otherwise, different implementations may terminate the fitting process at different points along the path, which would prevent a fair comparison of computing times.

Finally, let us compare these results to those presented in \citet{She2012}.  For the algorithm presented in that paper, the author reports an average time of 32 minutes to estimate the group SCAD regression coefficients when $n=100$ and $p=500$.  For the same size problem, our approach required a mere 0.35 seconds.

\section{Simulation studies}
\label{Sec:sim}

In this section, we compare the performance of group lasso, group MCP, and group SCAD using simulated data.  First, a relatively basic setting is used to illustrate the primary advantages of group MCP and group SCAD over group lasso.  We then attempt to mimic two settings in which the methodology might be used: to allow flexible semiparametric modeling of continuous variables and in genetic association studies, which involve large numbers of categorical variables.  We use the term ``null predictor'' to refer to a covariate whose associated regression coefficient is zero in the true model.

In all of the studies, five-fold cross-validation was used to choose the regularization parameter $\lam$.  Group SCAD and group MCP have an additional tuning parameter, $\gamma$.  In principle, one may attempt to select optimal values of $\gamma$ using, for example, cross-validation over a two-dimensional grid or using an information criterion approach.  Here, we fix $\gamma=3$ for group MCP and $\gamma=4$ for group SCAD, roughly in line with the default recommendations suggested in \citet{Fan2001} and \citet{Zhang2010} in the non-grouped case.

In Section~\ref{Sec:sim-basic}, we evaluate model accuracy by root mean squared error (RMSE):
\as{\textrm{RMSE} = \sqrt{\frac{1}{p}\sum_{j,k}(\beta_{jk}-\bh_{jk})^2}.}
In Sections~\ref{Sec:sim-cont} and \ref{Sec:sim-cat}, because the model fit to the data is not always the same as the generating model, we focus on root model error (RME) instead:
\as{\textrm{RME} = \sqrt{\frac{1}{n}\sum_i(\mu_i-\mh_i)^2},}
where $\mu_i$ and $\mh_i$ denote the true and estimated mean of observation $i$ given $\x_i$.  Note that the model error, which is also discussed in \citet{Fan2001}, is equal in expected value to the prediction error minus the irreducible error $\sigma^2$.  In all simulations, errors follow a standard Gaussian distribution and results are averaged over 1,000 independently generated data sets.

\subsection{Basic}
\label{Sec:sim-basic}

We begin with a very straightforward study designed to illustrate the basic shortcomings of the group lasso in comparison with group MCP and group SCAD.  The design matrix consists of 100 groups, each with 4 elements.  In five of these groups, the coefficients are $\pm \beta$; in the others, the true regression coefficients are zero.  Covariate values were generated from the standard normal distribution.  We fixed the sample size at 100 ({\em i.e.}, n=100, p=400) and varied $\abs{\beta}$ between 0 and 1.5.  In principle, group lasso should struggle when $\abs{\beta}$ is large, as it cannot alleviate the problem of bias towards zero for large coefficients without lowering $\lam$ and thereby allowing null predictors to enter the model.  Indeed, as Figure~\ref{Fig:basic} illustrates, this is exactly what occurs.

\begin{figure}
 \centering
 \includegraphics[width=\linewidth]{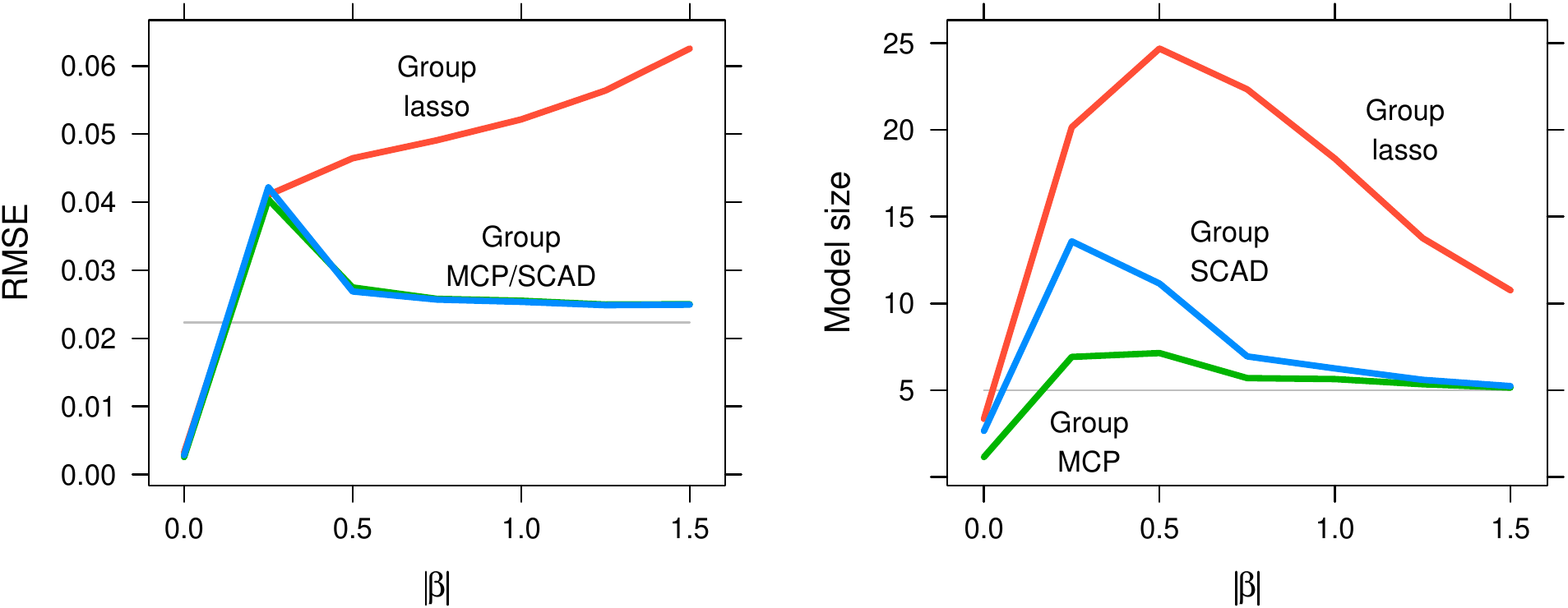}
 \caption{\label{Fig:basic} The impact of increasing coefficient magnitude on group regularization methods.  Model size is given in terms of number of groups ({\em i.e.}, the number of variables in the model is four times the amount shown).  The faint gray line on the left is the theoretically optimal RMSE than can be achieved in this setting.  The faint gray line on the right is the true model size.}
\end{figure}

For small values of the regression coefficients, all three group regularization methods perform similarly.  As we increase the magnitude of these coefficients, however, group MCP and group SCAD begin to estimate $\bb$ with an error approaching the theoretically optimal value, while group lasso performs increasingly poorly.  Furthermore, group MCP and group SCAD select much smaller models and approach the true model size much faster than group lasso, which selects far too many variables.

Comparing group MCP and group SCAD, the two are nearly identical in terms of estimation accuracy.  However, group MCP selects a considerably more sparse model, and has better variable selection properties.  Thus, although the two methods behave similarly in an asymptotic setting, group MCP seems to have somewhat better finite-sample properties.

In the left panel of Figure~\ref{Fig:basic}, we include an ``oracle'' RMSE for reference; we now clarify what exactly we mean by this.  An oracle model (one which knows in advance which coefficients are zero and which are nonzero) would be able to achieve a mean squared error of zero for the zero-coefficient variables and a total MSE of $\textrm{tr}\{(\X^T\X)^{-1}\}$ for the nonzero variables.  In our simulation $\X$ was random, with $\Ex(\X\X^T)=\I$.  Thus, the total MSE of the oracle model is approximately $\textrm{tr}(\I_0/n)$ where $\I_0$ is the identity matrix with dimension equal to that of the nonzero coefficients, with RMSE $\sqrt{s/n}$, where $s$ is the sparsity fraction ({\em i.e.}, the fraction of coefficients that are nonzero).  Note that in any finite sample, the columns of $\X$ will be correlated, so even the oracle model cannot achieve this RMSE for finite sample sizes; the gray line in Figure~\ref{Fig:basic} is thus the optimal RMSE that could be theoretically be achieved in this setting.

\subsection{Semiparametric regression}
\label{Sec:sim-cont}

Our next simulation involves groups of covariates constructed by taking basis expansions of continuous variables to allow for flexible covariate effects in semiparametric modeling.  The sample size was 200 and the data consisted of 100 variables, each of which were generated as independent uniform (0,1) variates.  The first six variables had potentially nonlinear effects given by the following equations:
\as{f_1(x) &= 2(e^{-10x}-e^{-10})/(1-e^{-10})-1\\
  f_2(x) &= -2(e^{-10x}-e^{-10})/(1-e^{-10})+1\\
  f_3(x) &= 2x-1\\
  f_4(x) &= -2x+1\\
  f_5(x) &= 8(x-0.5)^2-1\\
  f_6(x) &= -8(x-0.5)^2+1;
}
the other 94 had no effect on the outcome.  The scaling of the functions is to ensure that each variable attains a minimum and maximum of $(-1,1)$ over the domain of $x$ and thus that the effects of all six variables are roughly comparable.  Visually the effect of the six nonzero variables is illustrated below:
\begin{center}
\includegraphics[width=0.75\linewidth]{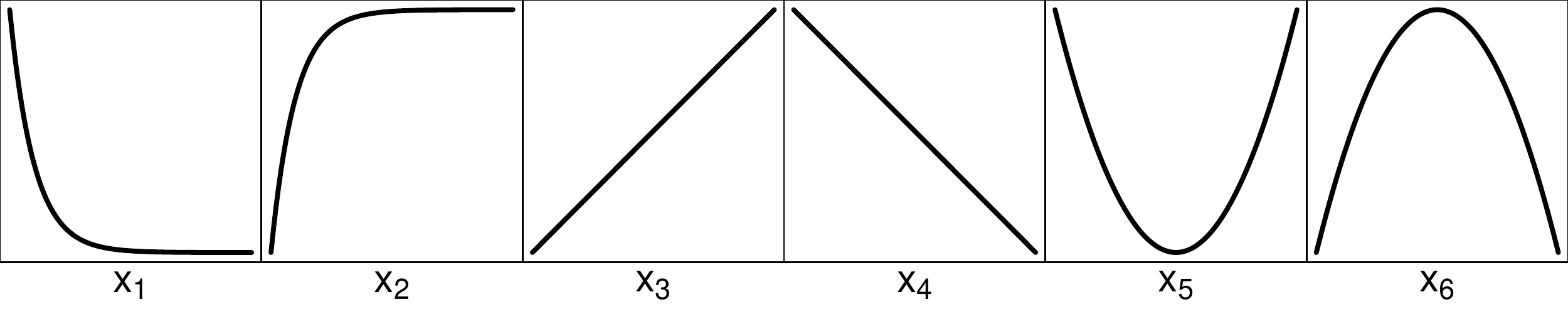}
\end{center}
For model fitting, each variable was represented using a 6-term B-spline basis expansion ({\em i.e.}, $\X$ had dimensions $n=200$, $p=600$).  In addition to the three group selection methods, models were also fit using the lasso ({\em i.e.}, ignoring grouping).  Results are given in Table~\ref{Tab:cont}.

\ctable[
 caption = Prediction and variable selection accuracy for the semiparametric regression simulation.,
 label = Tab:cont,
 pos=ht!]{lrrrr}{
 \tnote[a]{$\SE \leq 0.002$}
 \tnote[b]{$\SE \leq 0.4$}}{
\toprule
 & & {\bf Group} & {\bf Group} & {\bf Group}\\
& {\bf Lasso} & {\bf Lasso} & {\bf MCP} & {\bf SCAD} \\
\cmidrule{1-5}
RME\tmark[a] & 0.73 & 0.59 & 0.50 & 0.52\\
Variables selected\tmark[b] & 31.5 & 29.3 & 10.4 & 23.1 \\
\bottomrule}

Certainly, all three group selection approaches greatly outperform the lasso here.  However, as in the previous section, group MCP and group SCAD are able to achieve superior prediction accuracy which selecting more parsimonious models.  Also as in the previous simulation, group MCP and group SCAD perform similarly as far as prediction accuracy, but group MCP is seen to have better finite-sample variable selection properties --- recall that the true number of variables in the model is only 6.

\subsection{Genetic association study}
\label{Sec:sim-cat}

Finally, we carry out a simulation designed to mimic a small genetic association study involving single nucleotide polymorphisms (SNPs).  Briefly, a SNP is a point on the genome in which multiple versions (alleles) may be present.  A SNP may take on three values $\{0,1,2\}$, depending on the number of minor alleles present.  The effect is not necessarily linear --- for example, if the allele has a recessive effect, the phenotype associated with $x=0$ and $x=1$ are identical, while the phenotype associated with $x=2$ is different.  In such studies, it is desirable to have a method which is robust to different mechanisms of action, yet powerful enough to actually detect important SNPs, as the number of SNPs is typically rather large.

We simulated data involving 250 subjects and 500 SNPs, each of which was represented with 2 indicator functions ({\em i.e.}, n=250, p=1,000).  Three of the variables had an effect on the phenotype (one dominant effect, one recessive effect, one additive effect); the other 497 did not.  In addition to the three group selection methods, we included for comparison two versions of the lasso: one applied to all $p=1,000$ variables and ignoring grouping, the other assuming an additive effect for each genotype.  Note that for the second approach, $p=500$ as we estimate only a single coefficient for each SNP.  The results of this simulation are given in Table~\ref{Tab:cat}.

\ctable[
 caption=Prediction and variable selection accuracy for the genetic association study simulation.  True discoveries are selected variables that have a truly nonzero effect; false discoveries are selected SNPs that have no effect on the phenotype.,
 label=Tab:cat]{lrrrrr}{
 \tnote[a]{$\SE \leq 0.002$}
 \tnote[b]{$\SE \leq 0.02$}
 \tnote[c]{$\SE \leq 0.6$}}{\FL
 & & {\bf Additive} & {\bf Group} & {\bf Group} & {\bf Group}\\
& {\bf Lasso} & {\bf Lasso} & {\bf Lasso} & {\bf MCP} & {\bf SCAD} \\
\ML
RME\tmark[a]             & 0.38 & 0.37 & 0.34 & 0.25 & 0.27 \\
True discoveries\tmark[b]  &  2.6 & 2.2 & 2.6 & 2.4 & 2.5 \\
False discoveries\tmark[c] & 21.1 & 16.3 & 15.5 & 4.0 & 12.5 \LL}

The same broad conclusions may be reached here as in the previous simulations.  In particular, we note that (1) The group selection methods outperform the variable selection methods that either do not account for grouping or that attempt to incorporate grouping in an ad-hoc fashion.  (2) Group MCP and group SCAD outperform the group lasso both in terms of prediction accuracy as well as the number of false discoveries.  (3) Although group MCP and group SCAD are similar in terms of prediction accuracy, group MCP has significantly better variable selection properties, producing only four false discoveries compared to group SCAD's 12.5.

\section{Real data}

We give two examples of applying grouped variable selection methods to real data.  The first is a gene expression study in rats to determine genes associated with Bardet-Biedl syndrome.  The second is a genetic association study to determine SNPs associated with age-related macular degeneration.  As in the previous section, we fix $\gamma=3$ for group MCP, $\gamma=4$ for group SCAD, and select $\lam$ via cross-validation.  For the continuous outcome in Section~\ref{Sec:scheetz}, we use root cross-validation error, defined analogously to root model error in Section~\ref{Sec:sim}, to evaluate predictive accuracy.  For the binary outcome in Section~\ref{Sec:gas}, we use cross-validated misclassification error.

\subsection{Bardet-Biedl syndrome gene expression study}
\label{Sec:scheetz}

The data we analyze here is discussed more fully in \citet{Scheetz2006}.  Briefly, the data set consists of normalized microarray gene expression data harvested from the eye tissue of 120 twelve-week-old male rats.  The outcome of interest is the expression of TRIM32, a gene which has been shown to cause Bardet-Biedl syndrome \citep{Chiang2006}.  Bardet-Biedl syndrome is a genetic disease of multiple organ systems including the retina.

Following the approach in \citet{Scheetz2006}, 18,976 of the 31,042 probe sets on the array ``exhibited sufficient signal for reliable analysis and at least 2-fold variation in expression.''  These probe sets include TRIM32 and 18,975 other genes that potentially influence its expression.  We further restricted our attention to the 5,000 genes with the largest variances in expression (on the log scale) and considered a three-term natural cubic spline basis expansion of those genes, resulting in a grouped regression problem with $n=120$ and $p=15,000$.  The models selected by group lasso, group MCP, and group SCAD are described in Table~\ref{Tab:scheetz}.

\ctable[
 caption={Genes selected by group lasso/SCAD/MCP, along with the Euclidean norm of the coefficients for each gene's basis expansion.},
 label=Tab:scheetz,
 pos=ht!]{llrrr}{}{
\toprule
 & & \multicolumn{3}{c}{Group norm} \\
 \cline{3-5}
{\bf Probe} & {\bf Gene}   & {\bf Group} & {\bf Group} & {\bf Group} \\
{\bf Set}   & {\bf Symbol} & {\bf Lasso} & {\bf MCP} & {\bf SCAD} \\
\midrule
  1374131\_at &  & 0.11 &  & 0.13 \\
  1383110\_at & Klhl24 & 0.10 &  & 0.08 \\
  1383749\_at & Phospho1 & 0.02 &  & 0.04 \\
  1376267\_at &  & 0.22 &  & 0.23 \\
  1377791\_at &  & 0.13 &  & 0.12 \\
  1376747\_at &  & 0.28 &  & 0.24 \\
  1390539\_at &  & 0.11 &  & 0.12 \\
  1384470\_at &  & 0.05 &  & 0.07 \\
  1386032\_at & Prkd3 & 0.03 &  & 0.06 \\
  1393231\_at & Ppp4r2 & 0.01 &  & 0.03 \\
  1385798\_at &  & 0.03 &  &  \\
  1383730\_at & Ttc9c & 0.02 &  & 0.06 \\
  1368476\_at & Nr3c2 & 0.05 &  & 0.01 \\
  1384860\_at & Zfp84 & 0.03 &  & 0.07 \\
  1372928\_at &  & 0.22 & 1.83 & 0.20 \\
  1381902\_at & Zfp292 & 0.16 &  & 0.18 \\
  1390574\_at &  & 0.02 &  & 0.01 \\
  1384940\_at & Zfp518a & 0.10 &  & 0.10 \\
\bottomrule}

This is an interesting case study in that group MCP selects a very different model from the other two approaches.  In particular, group lasso and group SCAD each select a fairly large number of genes, while shrinking each gene's group of coefficients nearly to zero.  Group MCP, on the other hand, selects a single gene and returns a fit nearly the same as the least-squares fit for that gene alone.  The relationship between probe set 1372928\_at and TRIM32 estimated by each model is plotted in Figure~\ref{Fig:scheetz}.

\begin{figure}
 \centering
 \includegraphics[width=0.75\linewidth]{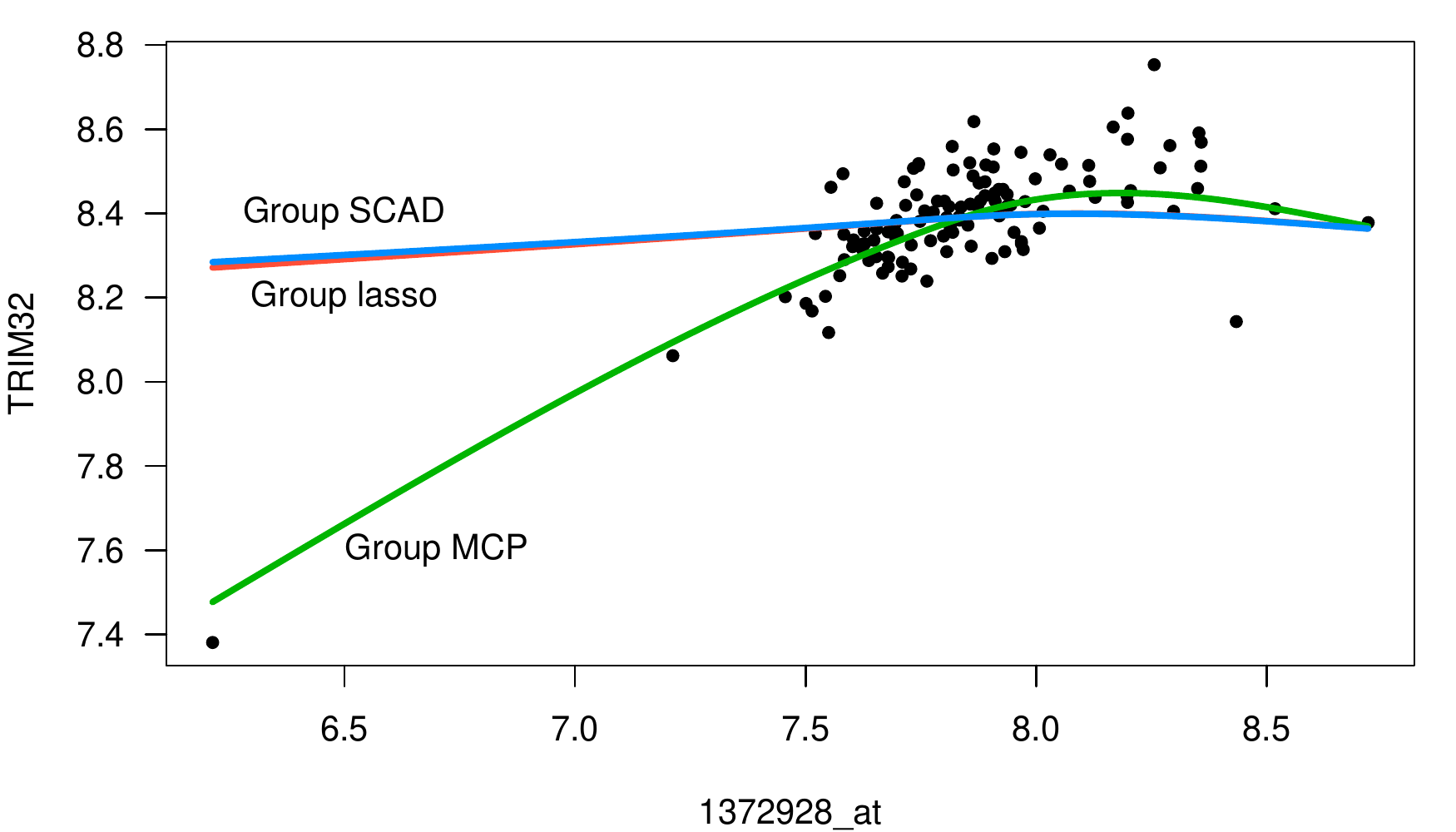}
 \caption{\label{Fig:scheetz} Estimated relationship between probe set 1372928\_at and TRIM32 estimated by group lasso, group MCP, and group SCAD.  Estimates are superimposed on top of a scatterplot and restricted to pass through the mean expression for each probe set.}
\end{figure}

The most important aspect of Figure~\ref{Fig:scheetz} to note is that the outcome, TRIM32, has a large outlying value almost 10 standard deviations below the mean of the rest of the points.  This observation has a large impact on the fit: for all of the genes in Table~\ref{Tab:scheetz}, this subject is also responsible for the lowest/highest or second-lowest/highest value of that gene, and the median absolute correlation for the genes in the table is 0.62.  Many of the scatterplots of those genes versus TRIM32 look qualitatively similar to the one in Figure~\ref{Fig:scheetz}.

Faced with this set of genes, group MCP selects a single gene and fits a model that explains 65\% of the variance in TRIM32 expression.  Group SCAD and group lasso select an ensemble of correlated genes, downweighting the contribution of each gene considerably.  Each approach has advantages, depending on the goal of the analysis.  The group SCAD/lasso approaches avoid a possibly arbitrary selection of one gene from a highly correlated set, and produce a model with somewhat better predictive ability (root cross-validation error of $0.092 \pm 0.04$ versus $0.099 \pm 0.04$), although all three approaches are within random variability of each other.  Group MCP, on the other hand, produces a highly parsimonious model capable of predicting just as well as the group lasso/SCAD models despite using only a single gene.  This is potentially a valuable property if the goal is, say, to develop a diagnostic assay and each gene that need to be measured adds cost to the assay.


Finally, it should be noted that although group MCP is behaving in a rather greedy fashion in this example, this is not an inherent aspect of the method.  By adjusting the $\gamma$ parameter, group MCP can be made to resemble the group lasso and group SCAD solutions as well -- recall that group lasso can be considered a special case of group MCP with $\gamma=\infty$.  Group SCAD, on the other hand, cannot be made to resemble group MCP and is incapable in this case of selecting a highly parsimonious model.  Group MCP is considerably more flexible, although of course to pursue this flexibility, proper selection of tuning parameters becomes an important issue.  The selection of additional tuning parameters is an important area for further study in both the grouped and non-grouped application of the MC penalty.

\subsection{Age-related macular degeneration genetic association study}
\label{Sec:gas}

We analyze data from a case-control study of age-related macular degeneration consisting of 400 cases and 400 controls.  We confine our analysis to 532 SNPs that previous biological studies have suggested may be related to the disease.  As in Section~\ref{Sec:sim-cat}, we represent each SNP as a three-level factor depending on the number of minor alleles the individual carries at that locus.  This requires the expansion of each SNP into a group of two indicator functions; our design matrix for this example thus has $n=800$ and $p=1,064$.  Group regularized logistic regression models were then fit to the data; as before, $\lam$ was selected via cross-validation.  The results are presented in Table~\ref{Tab:mullins}.

\ctable[
 caption={Application of group regularized methods to age-related macular degeneration study.  The number of SNPs selected by the method as well as the cross-validated misclassification error are reported.  For comparison, the intercept-only model is also listed (``Baseline'').},
 label=Tab:mullins,
 pos=ht!]{lcc}{
 \tnote{$\SE = 0.02$ for all models}}{
\toprule
 & & {\bf Misclassification} \\
 & {\bf SNPs} & {\bf Error\tmark} \\
\midrule
Baseline &    & 0.50 \\ 
grLasso & 51 & 0.41 \\ 
grMCP   & 12 & 0.39 \\ 
grSCAD  & 15 & 0.41 \LL}

We make the following remarks on Table~\ref{Tab:mullins}: (1) All three approaches represent significant improvements over the baseline (intercept-only) misclassification error.  (2) The group MCP model is slightly superior to the group lasso and group SCAD models, although again, the difference between the three models in terms of predictive accuracy is comparable to the SE.  (3) The group MCP and group SCAD approaches produce considerably more parsimonious models without a loss in prediction accuracy.  Compared with the gene expression example of the previous section, parsimony is both more desirable and more believable in this example.  The SNPs selected here are not highly correlated
and thus more likely to represent independent causes than dependent manifestations of a separate underlying cause such as up-regulation of a pathway.

\section{Conclusion}

Group MCP and group SCAD models are powerful alternatives to the group lasso in problems involving grouped variable selection.  However, application and study of these approaches has been limited, especially in high-dimensional problems, due to a lack of efficient algorithms and a lack of publicly available software for fitting these models.  In this article, we attempt to remedy this, describing the development of efficient algorithms and proving an implementation via the \texttt{R} package \texttt{grpreg}.

\section*{Acknowledgments}

Jian Huang's work is supported in part by NIH Grants R01CA120988, R01CA142774 and NSF Grants DMS-0805670 and DMS-1208225.  The authors would like to thank Rob Mullins for the genetic association data analyzed in Section~\ref{Sec:gas}, as well as the associate editor and two anonymous reviewers, who provided many helpful remarks that led to considerable refinement of this article.

\section*{Appendix}

Before proving Proposition~\ref{Prop:linear}, we establish the groupwise convexity of all the objective functions under consideration.  Note that for group SCAD and group MCP, although they contain nonconvex components and are not necessarily convex overall, the objective functions are still convex with respect to the variables in a single group.

\begin{lemma}
\label{Lemma:convex-lin}
The objective function $Q(\bb_j)$ for regularized linear regression is a strictly convex function with respect to $\bb_j$ for the group lasso, for group SCAD with $\gamma > 2$, and for group MCP with $\gamma > 1$.
\end{lemma}

\begin{proof} Although $Q(\bb_j)$ is not differentiable, it is directionally twice differentiable everywhere.  Let $\nabla_{\d}^2 Q(\bb_j)$ denote the second derivative of $Q(\bb_j)$ in the direction $d$.  Then the strict convexity of $Q(\bb_j)$ follows if $\nabla_{\d}^2 Q(\bb_j)$ is positive definite at all $\bb_j$ and for all $\d$.  Let $\xi_*$ denote the infimum over $\bb_j$ and $\d$ of the minimum eigenvalue of $\nabla_{\d}^2 Q(\bb_j)$.  Then, after some algebra, we obtain
\as{
  \xi_* &= 1 & \hspace{-1in} & \textrm{Group lasso}\\
  \xi_* &= 1 - \frac{1}{\gamma-1} & \hspace{-1in} & \textrm{Group SCAD}\\
  \xi_* &= 1 - \frac{1}{\gamma} & \hspace{-1in} & \textrm{Group MCP},
}
These quantities are positive under the conditions specified in the lemma.
\end{proof}

We now proceed to the proof of Proposition~\ref{Prop:linear}.

\begin{proof}[Proof of Proposition~\ref{Prop:linear}]
The descent property is a direct consequence of the fact that each updating step consists of minimizing $Q(\bb)$ with respect to $\bb_j$.  Lemma~\ref{Lemma:convex-lin}, along with the fact that the least squares loss function is continuously differentiable and coercive, provide sufficient conditions to apply Theorem 4.1 of \citet{Tseng2001}, thereby establishing that every limit point of $\{\bb^{(m)}\}$ is a stationary point of $Q(\bb)$.

We further note that $\{\bb^{(m)}\}$ is guaranteed to converge to a unique limit point.  Suppose that the sequence possessed two limit points, $\bb'$ and $\bb''$, such that for at least one group, $\bb'_j \neq \bb''_j$.  For the transition $\bb' \to \bb''$ to occur, the algorithm must pass through the point $(\bb''_j, \bb'_{-j})$.  However, by Lemma 1, $\bb'_j$ is the unique value minimizing $Q(\bb_j|\bbj)$.  Thus, $\bb' \to \bb''$ is not allowed by the group descent algorithm and $\{\bb^{(m)}\}$ possesses a single limit point.
\end{proof}

For Proposition~\ref{Prop:logistic}, involving logistic regression, we proceed similarly, letting $R(\bb|\tbb)$ denote the majorizing approximation to $Q(\bb)$ at $\tbb$.

\begin{lemma}
\label{Lemma:convex-log}
The majorizing approximation $R(\bb_j|\tbb)$ for regularized logistic regression is a strictly convex function with respect to $\bb_j$ at all $\tbb$ for the group lasso, for group SCAD with $\gamma > 5$, and for group MCP with $\gamma > 4$.
\end{lemma}

\begin{proof} Proceeding as in the previous lemma, and letting $\xi_*$ denote the infimum over $\tbb$, $\bb_j$ and $\d$ of the minimum eigenvalue of $\nabla_{\d}^2 Q(\bb_j)$, we obtain
\as{
  \xi_* &= \frac{1}{4} & \hspace{-1in} & \textrm{Group lasso}\\
  \xi_* &= \frac{1}{4} - \frac{1}{\gamma-1} & \hspace{-1in} & \textrm{Group SCAD}\\
  \xi_* &= \frac{1}{4} - \frac{1}{\gamma} & \hspace{-1in} & \textrm{Group MCP},
}
These quantities are positive under the conditions specified in the lemma.
\end{proof}

\begin{proof}[Proof of Proposition~\ref{Prop:logistic}]
The proposition makes two claims: descent with every iteration and convergence to a stationary point.  To establish descent for logistic regression, we note that because $L$ is twice differentiable, for any point $\be$ there exists a vector $\be^{**}$ on the line segment joining $\be$ and $\be^*$ such that
\as{
  L(\be) &= L(\be^*) + (\be-\be^*)^T \nabla L(\be^*) + \frac{1}{2} (\be-\be^*)^T \nabla^2 L(\be^{**}) (\be-\be^*)\\
         &\leq \tilde{L}(\be|\be^*)
}
where the inequality follows from the fact that $v\I - \nabla^2 L(\be^{**})$ is a positive semidefinite matrix.  Descent now follows from the descent property of MM algorithms \citep{Lange2000} coupled with the fact that each updating step consists of minimizing $R(\bb_j|\tbb)$.

To establish convergence to a stationary point, we note that if no elements of $\bb$ tend to $\pm\infty$, then the descent property of the algorithm ensures that the sequence $\bb^{(k)}$ stays within a compact set and therefore possesses a limit point $\tbb$.  Then, as in the proof of Proposition~\ref{Prop:linear}, Lemma~\ref{Lemma:convex-log} allows us to apply the results of \citet{Tseng2001} and conclude that $\tbb$ must be a stationary point of $R(\bb|\tbb)$.  Furthermore, because $R(\bb|\tbb)$ is tangent to $Q(\bb)$ at $\tbb$, $\tbb$ must also be a stationary point of $Q(\bb)$.
\end{proof}

\bibliographystyle{ims-nourl}

\end{document}